\documentclass[12pt]{article}
\parskip=\medskipamount                 
\thicklines                     
\usepackage{amsfonts,amssymb,amsmath,amscd, amsthm}
\usepackage{bbm}
\usepackage{booktabs}

\renewcommand{\tilde}{\widetilde}

\newcommand{\ZZ}{{\mathbb Z}}

\newcommand{\cA}{{\mathcal A}}

\newcommand{\cF}{{\mathcal F}}

\newcommand{\psip}{\psi_+}
\newcommand{\hpsip}{\widehat \psi_+}
\newcommand{\hphi}{{\hat \phi}}
\newcommand{\hX}{{\widehat X}}
\newcommand{\hSigma}{{\widehat \Sigma}}
\newcommand{\hPhi}{{\widehat \Phi}}
\newcommand{\hg}{{\widehat g}}
\newcommand{\hb}{{\widehat b}}
\newcommand{\hpsi}{{\widehat \psi}}
\newcommand{\hD}{{\widehat D}}
\newcommand{\hf}{{\widehat f}}
\newcommand{\homega}{{\widehat \omega}}
\newcommand{\hj}{{\widehat j}}

\newcommand{\T}{{\mathsf T}}
\newcommand{\TYperp}{{(TY)^\bot}}
\newcommand{\TM}{{e^*TM}}
\newcommand{\Rp}{{\mathcal R_+}}

\newtheorem{definition}{Definition}
\newtheorem{prop}{Proposition}
\newtheorem*{thm*}{Theorem}
\newtheorem*{cor*}{Corollary}
\newtheorem*{definition*}{Definition}
\newtheorem*{prop*}{Proposition}

\font\teneusm=eusm10
\font\seveneusm=eusm7
\font\fiveeusm=eusm5
\newfam\eusmfam
\textfont\eusmfam=\teneusm
\scriptfont\eusmfam=\seveneusm
\scriptscriptfont\eusmfam=\fiveeusm

\begin{document}

\title{Geometry of Topological Defects of Two-dimensional Sigma Models}

\author{Anton Kapustin, Kevin Setter \\ {\it \small California
Institute of Technology}}

\begin{titlepage}
\maketitle

\abstract{A topological defect separating a pair of two-dimensional CFTs is a codimension one interface along which all components of the stress-energy tensor glue continuously.  We study topological defects of the bosonic, (0,1)- and (0,2)-supersymmetric sigma models in two dimensions.  We find a geometric classification of such defects closely analogous to that of A-branes of symplectic manifolds, with the role of symplectic form played instead by a neutral signature metric.  Alternatively, we find a  compact description in terms of a generalized metric on the product of the targets.  In the (0,1) case, we describe the target space geometry of a bundle in which the fermions along the defect take values.  In the (0,2) case, we describe the defects as being simultaneously A-branes and B-branes.}

\end{titlepage}

\section{Introduction} \label{sec:intro}

It is well-known that two-dimensional sigma models with differing amounts of left and right-moving supersymmetry (e.g. the (0,1) and (0,2) supersymmetric sigma models) cannot be consistently defined on worldsheets with nontrivial boundary.  This is essentially because a boundary condition must set to zero a linear combination of left- and right-moving fermions.  

This does not however rule out the possibility of defining these models on a worldsheet equipped with \emph{defects}, i.e. one dimensional submanifolds $D$ along which the fields of two distinct CFTs are glued together consistently.  (A boundary condition is a special case in which one of the CFTs is trivial.)  Our goal is to study supersymmetric defects in the (0,1) and (0,2) sigma models.    

Defects of two dimensional CFTs have attracted much interest recently, especially in the context of rational conformal field theory and as a means of elegantly implementing dualities \cite{Kapustin:2009av}, \cite{Davydov:2010rm}, \cite{Sarkissian:2008dq}, \cite{Fuchs:2007tx}.  In this paper, we will analyze criteria for  preserving superconformal symmetry, disregarding any enhanced chiral symmetry which the theories may possess; moreover, we will focus on the target space geometry of defects, and shall only briefly comment on their role in implementing dualities.   

Consider a two-dimensional worldsheet disconnected into two domains $\Sigma$ and $\hSigma$ by a defect line $D$.  On the domain $\Sigma$, one defines a sigma model of maps $\Phi : \Sigma \to X$ and on the domain $\hSigma$ one defines a sigma model of maps $\hPhi : \hSigma \to \hX$, where $X$ and $\hX$ are two compact, Riemannian target spaces.  The restriction of these maps to $D$ defines a product map $\Phi \times \hPhi |_D : D \to X \times \hX$.  As observed in \cite{Fuchs:2007fw}, gluing conditions will require the product map to takes values in some submanifold $Y \subseteq X \times \hX$ of the product of the targets.  

One may also include a term in the action coupling the bulk fields to a line bundle with connection living on the worldvolume of $Y$ (equipping it with a closed 2-form $F$), exactly by analogy with Chan-Paton bundles on branes for the case of sigma model boundary conditions.  Borrowing the terminology of \cite{Fuchs:2007fw}, we refer to the pair $(Y,F)$ as a \emph{bibrane}.

To ensure that the defect theory preserves a specified set of symmetries of the bulk theories, we impose certain requirements on the geometry of the embedding $Y$, as well as on the choice of the 2-form $F$.  Let us describe the relevant set of symmetries we wish to preserve by listing the associated No\"ether charges. The (0,1) supersymmetry algebra in 1+1 dimensions reads
\[ 
  \{Q_+,Q_+\} = H + P
\]
where $Q_+$ is the single, real, right-moving supercharge, $H$ is the worldsheet energy, and $P$ is the worldsheet momentum.  We wish to find $(Y,F)$ such that  $H$, $P$, and $Q_+$ remain as conserved charges of the defect theory; since $H$ and $P$ will be separately conserved, we will be studying examples of what are known in the literature as \emph{topological defects} \cite{Fuchs:2007tx}. 

Already in the case of the bosonic sigma model (i.e. with no fermionic fields present) it is an interesting question which bibranes $(Y,F)$ define topological defects, and it is to this that we will first direct our attention.  It turns out to be natural to choose the neutral signature metric $G = g \oplus - \hg$ on the (pseudo-Riemannian) product manifold $M = X \times \hX$, where $g$ and $\hg$ are the positive-definite metrics on $X$ and $\hX$.

Those $(Y,F)$ that supply topological defects turn out to bear a structural resemblance to A-branes of symplectic manifolds (exchanging the symplectic form on the ambient target space for the neutral signature metric defined above); for instance, just as A-branes are coisotropic submanifolds with respect to the symplectic form on the target, $Y$ will be required to be ``coisotropic" with respect to $G$, in a sense that we will explain in section \ref{sec:geometry}.  Moreover, we will describe two special classes of topological defects (the ``graphs of isometries" and ``half para-K\"ahler" defects), which are the analogs of Lagrangian and space-filling A-branes. 

Alternatively, it turns out to be natural to employ the language of Hitchin's generalized geometry \cite{Gualtieri:2007ng}, in terms of which we obtain the following simple characterization:  $(Y,F)$ define a topological defect of the bosonic sigma model if and only if the ``$F$-rotated generalized tangent bundle" of $Y$ is stabilized by the ``generalized metric" on $X \times \hX$.  We will explain what these terms mean in greater depth in section \ref{sec:geometry}.     

We refer to $(Y,F)$ satisfying the above stabilization condition as \emph{topological bibranes}; they can be defined for any neutral signature manifold $M$ and are, perhaps, mathematically interesting in their own right.  However, we should point out that the neutral signature manifolds relevant to the current physics discussion are of a very restricted type: namely, manifolds that can be expressed as a global product $M = X \times \hX$ such that the metric restricted to $TX$ (resp. $T \hX$) is positive (resp. negative) definite.

In section \ref{sec:gluing} we discuss defect gluing conditions in general and say what it means for a gluing condition to preserve a symmetry of the bulk theories.  In section \ref{sec:bosonic} we write down gluing conditions on the fields corresponding to the choice of $(Y,F)$ and analyze the topological defect requirement.  In section \ref{sec:geometry}, we explain the analogy with A-branes and reformulate the topological bibrane condition on $(Y,F)$ terms of generalized geometry.  In section \ref{sec:susy01} we analyze defects of the (0,1) supersymmetric sigma model, supplementing the bosonic gluing conditions with an additional fermionic gluing condition and studying the geometry of a certain middle dimensional subbundle of $TY$.   In section \ref{sec:susy02} we treat topological defects of the (0,2) sigma model, which we will describe as those $(Y,F)$ that are simultaneously A-branes and B-branes with respect to a certain symplectic form and complex structure.   Finally, in section \ref{sec:tduality} we briefly comment on a subset of our topological defects which implement dualities relating the sigma model theories on $\Sigma$ and $\hSigma$.  

Proofs of selected propositions discussed in the text are offered in the appendix. 

K.S. thanks Ketan Vyas for useful discussions.  This work was supported in part by the DOE grant DE-FG02-92ER-40701.

\section{Defect gluing conditions} \label{sec:gluing}

Before discussing defects in specific theories, let us discuss defect gluing conditions in general and say what it means for a defect to preserve a symmetry of the bulk theories.

The variation of the action will, in general, consist of three types of terms:
\[
  \delta S = (\delta S)_\Sigma + (\delta S)_\hSigma + (\delta S)_D
\]
where $ (\delta S)_\Sigma$ and $(\delta S)_\hSigma$ are integrals over the domains $\Sigma$ and $\hSigma$ respectively, which vanish for field configurations solving the bulk equations of motion, and the third term, $(\delta S)_D$, is an integral over $D$, which in general, will not vanish unless we impose a gluing condition on the fields along $D$.  A gluing condition constrains the values of the fields along $D$ and hence also the set of \emph{allowed variations}, by which we mean variations mapping one solution of the gluing condition to another solution.  Therefore, the first requirement of a good gluing condition is that it sets $(\delta S)_D$ to zero identically for all allowed variations and that it does so ``minimally" (i.e. without overconstraining the data along $D$).

Additionally, we may wish for the defect theory to preserve a certain symmetry of the bulk.  Let $(\sigma^0, \sigma^1)$ be worldsheet coordinates in which $D$ is described locally as the set of points with $\sigma^1 = 0$.  In these coordinates, we say that a gluing condition classically preserves a symmetry of the bulk if and only if the 1- component of the associated No\"ether current glues continuously across $D$, thereby ensuring the existence of a conserved No\"ether charge in the composite theory.  (The quantum theory may develop an anomaly, but we will confine our discussion to the classical problem.)  This being satisfied, the symmetry variation will automatically be among the allowed variations. 

For instance, conservation of the worldsheet energy $H$ requires that the off-diagonal components of the stress-energy tensors glue continuously:
\[
  T^{\, 1}_{\;\;\; 0} - {\widehat T}^{\,1}_{\;\;\; 0} = 0
\]
at points of $D$.  Defects satisfying this condition are said to be \emph{conformal defects} \cite{Bachas:2001vj}.

If, in addition, the defect gluing condition ensures that the diagonal components of the stress-energy tensor glue continuously:
\[
  T^{\, 1}_{\;\;\; 1} - {\widehat T}^{\,1}_{\;\;\; 1} = 0
\]
then worldsheet momentum $P$ is also a conserved charge.  Defects satisfying this condition are called \emph{topological defects} \cite{Fuchs:2007tx},  due to the fact that the location of $D$ on the worldsheet can be deformed smoothly without affecting the values of correlators (so long as it does not cross through the location of a local operator insertion). 

In subsequent sections, we will write down an action, vary the action, and then systematically analyze what gluing conditions set $(\delta S)_D$ minimally and ensure continuous gluing across $D$ of the 1-components of relevant No\"ether currents.    

\section{Topological defects of the bosonic sigma model} \label{sec:bosonic}

To begin, we analyze topological defects of the bosonic sigma model.

Let us fix notation.  As above, let $(\sigma^0, \sigma^1)$ be worldsheet coordinates in which the defect line $D$ is given locally by $\sigma^1 = 0$, and the worldsheet metric is taken to be flat, with signature (-, +).  Let $\Sigma$ be the domain given by $\sigma^1 \geq 0$, and $\hSigma$ be the domain  given by $\sigma^1 \leq 0$.  The fields of the bosonic sigma model with defect consist of maps $\Phi : \Sigma \to X$ and $\hPhi :\hSigma \to \hX$; in terms of local coordinates $\phi^i$ on $X$ and $\hphi^i$ on $\hX$, we can describe these maps by functions $\phi^i (\sigma)$ for $\sigma^1 \geq 0$ and $\hphi^i (\sigma)$ for $\sigma^1 \leq 0$.  (Note that the target spaces are assumed to have the same dimensionality $n$ since we are looking for topological -- not merely conformal -- defects.)  Likewise, the product map $\Phi \times \hPhi |_D : D \to X \times \hX$ can be described by the functions $\phi^I  = (\phi^i, \hphi^i) |_{(\sigma^0, 0)}$. Indices $i,j$ range from $1,\dots n$ and $I,J$ range from $1,\dots 2n$.

The total action for the theory with defect is the sum of bulk terms and a term coupling the bulk fields to a connection on a rank one vector bundle living on the submanifold $Y \subseteq X \times \hX$:
\begin{align*}
  S &= \int_{\Sigma} d^2 \sigma \Big(-\frac{1}{2} g_{ij} \partial_\mu \phi^i \partial^\mu \phi^j 
  -\frac{1}{2} b_{ij} \epsilon^{\mu \nu} \partial_\mu \phi^i \partial_\nu \phi^j \Big) \\
    &+ \int_{\hSigma} d^2 \sigma \Big(-\frac{1}{2} \hg_{ij} \partial_\mu \hphi^i \partial^\mu \hphi^j  - \frac{1}{2} \hb_{ij} \epsilon^{\mu \nu} \partial_\mu \hphi^i \partial_\nu \hphi^j \Big) 
    + \int_{D} d \sigma^0 \cA_I \partial_0 \phi^I
\end{align*}
where $\cA = \cA_I (\phi) d \phi^I$ is the connection 1-form, $g_{ij} (\phi)$ and $b_{ij}(\phi)$ are the metric and B-field on $X$, and $\hg_{ij} (\phi)$ and $\hb_{ij}(\phi)$ are the metric and B-field on $\hX$.  For simplicity we assume that both $b$ and $\hb$ are closed 2-forms.

Varying this action  and picking out the term localized on $D$, one finds
\begin{align}
  (\delta S)_D = \int_{D} d \sigma^0 \Big( &-(g_{ij} \partial_1 \phi^i + b_{ij} \partial_0 \phi^i)\, \delta \phi^j + (\hg_{ij} \partial_1 \hphi^i + \hb_{ij} \partial_0 \hphi^i)\, \delta \hphi^j \notag \\
&- (\partial_I \cA_J - \partial_J \cA_I) \, \partial_0 \phi^I \, \delta \phi^J \; \;\Big) 
\label{eq:deltaS}
\end{align}
Let $\sigma = (\sigma^0, 0)$ be a point on $D$ and let us write down gluing conditions on the fields at $\sigma$ sufficient to ensure vanishing of this expression for all allowed variations.   

As mentioned in the introduction, a choice of gluing condition corresponds to a choice of a submanifold $Y \subseteq X \times \hX$.   Let $y = \Phi \times \hPhi |_D (\sigma)$ be the image of $\sigma$ under the product map.   We impose first the Dirichlet condition 
\[
   y \in Y
\]
If $k$ is the dimension of $Y$, with $0 \leq k \leq 2n$, then the above represents $2n-k$ independent  gluing constraints on the $2n$ bosonic fields $\phi^I$.  In order to ensure vanishing of \eqref{eq:deltaS}, we supplement these Dirichlet conditions with a set of $k$ Neumann conditions on the derivatives $\partial_1 \phi^I$.  It is easiest to describe these in terms of three target space tangent vectors $u,v,w \in T_y ( X \times \hX )$, with 
\begin{align*}
  u &= u^I \partial_I = (\partial_0 \phi^I ) \partial_I \\
  v &= v^I \partial_I = (\partial_1 \phi^I ) \partial_I \\
  w &= w^I \partial_I = (\delta \phi^I) \partial_I 
\end{align*}
where $\partial_I \equiv \frac{\partial}{\partial \phi^I}$.  Constraining points of $D$ to be mapped to the submanifold $Y$ requires $u \in TY$, the subspace of vectors tangent to $Y$; moreover, allowed variations $\delta \phi^I$ are those such that $w \in TY$ as well. 

Here and throughout, we shall only have occasion to discuss vectors that are evaluated at points of $Y$; hence, we regard all vectors as lying in the pullback bundle $\TM$, where $e:Y \to M = X \times \hX$ is the embedding map, and shall regard  $TY$ as a subbundle $TY \subseteq \TM$.  

Since the target space metrics $g$ and $\hg$ enter  $(\delta S)_D$ with opposite signs, we find it useful to define a \emph{neutral signature} metric and B-field on $X \times \hX$ by
\[
  G_{I J} = \begin{pmatrix} g_{ij} & 0 \\ 0 & -\hg_{ij} \end{pmatrix}, \quad
  B_{I J} = \begin{pmatrix} b_{ij} & 0 \\ 0 & -\hb_{ij} \end{pmatrix}
\]
Moreover, we write $F = -\cF - e^* B$ for closed 2-form on $Y$ obtained by combining the curvature $\cF = d\cA$ of the line bundle and the pullback of the B-field.  Using these definitions, we may state the Neumann conditions as follows:
\[
   G(u,w) = F(v,w)
\]
for all $w \in TY$.  As promised this represents $k$ independent conditions, one for each linearly independent tangent vector $w$.
\begin{definition}
Together, the pair $(u,v)$ satisfying the above gluing conditions are said to be an \emph{allowed pair} of tangent vectors.  
\end{definition}
Having written down good gluing conditions on the fields, let us now analyze what additional constraints the topological defect condition places on allowed pairs; this will constrain the choice of submanifold $Y$ as well as the choice of 2-form $F$ on its worldvolume.  First of all, the conformal defect condition $ T^{\, 1}_{\;\;\; 0} - {\widehat T}^{\,1}_{\;\;\; 0} = 0$ is equivalent to
\[
  G(u,v) = 0
\]
for all allowed pairs $(u,v)$.  This is an automatic consequence of the antisymmetry of $F$ since, if $(u,v)$ is an allowed pair, then $ G(u,v) = G(v,u) = F(u,u) = 0$, where we have applied the definition of allowed pair, setting $w = u$.  Therefore, $(Y,F)$ automatically defines a conformal defect of the bosonic sigma model.

Further requiring the topological condition  $T^{\, 1}_{\;\;\; 1} - {\widehat T}^{\,1}_{\;\;\; 1} = 0 $ is equivalent to requiring
\[
  G(u,u) = - G(v,v)
\]
for all allowed pairs $(u,v)$.  This condition turns out to be equivalent to an apparently stronger condition, as follows.
\begin{prop}  \label{prop:stronger}
If $(Y,F)$ are such that $G(u,u) = -G(v,v)$ for all allowed pairs $(u,v)$, then $v \in TY$ and $(v,u)$ is an allowed pair as well.
\end{prop}
(See appendix for the proof of this proposition and other proofs which are not immediate.)  The converse is true as well, since if  $(Y,F)$  satisfy the conditions of this proposition, then by setting $w=u$, one has $G(u,u) = F(v,u) = -F(u,v) = -G(v,v)$.   Hence, we arrive at the following:
\begin{prop} \label{prop:topbibrane}
Suppose $Y \subseteq X \times \hX$ is a submanifold and $F$ is a closed 2-form on its worldvolume.  Then $(Y,F)$ supplies a topological defect of the bosonic sigma model if and only if the following condition is met:
\begin{itemize}
\item[] If $(u,v)$ is a pair of vectors such that $u \in TY, v \in \TM$ and $G(v,w)=F(u,w)$ for all $w \in TY$, then
\item[] $v \in TY$ as well, and $G(u,w) = F(v,w)$ for all $w \in TY$. 
\end{itemize}
\end{prop}
\begin{definition}
A pair $(Y,F)$ satisfying the conditions of the preceding theorem are said to define a \emph{topological bibrane}. 
\end{definition}

\section{Geometry of topological bibranes} \label{sec:geometry}

Having obtained a characterization of topological bibranes $(Y,F)$ above, we reformulate this condition slightly to put it in a more understandable form.  

We have chosen to equip the manifold $M = X \times \hX$ with a neutral signature metric $G$; let us therefore record some basic facts from the theory of submanifolds of indefinite signature spaces.  See \cite{Bejancu:2006p569} for a more complete discussion. 
\begin{definition}
Let $Y$ be an embedded submanifold of a $2n$-dimensional pseudo-Riemannian manifold $M$ equipped with a nondegenerate, symmetric metric $G$.  The orthogonal subbundle $\TYperp$ is defined to the be the set of vectors that are $G$-orthogonal to all of $TY$:
\[
   \TYperp = \{ u \in \TM : G(u,v) = 0 \qquad \text{for all $v \in TY$}\}
\]
\end{definition}
If $TY$ is $k$-dimensional, then $\TYperp$ is $(2n-k)$-dimensional.  The main difference between the theory of pseudo-Riemannian submanifolds as compared with the theory of Riemannian submanifolds is that the restriction of the metric $G$ to $Y$ can develop degenerate directions.  Hence, the subbundle $\Delta = TY \cap \TYperp$ will in general be nontrivial.  Borrowing terminology from symplectic geometry, we define the following three special classes of submanifolds:
\begin{definition} 
Depending on whether the bundle $TY$, regarded as a subbundle of $\TM$, contains (or is contained by) its orthogonal bundle $\TYperp$, we say that  
\begin{itemize}
\item[] $TY$ is \emph{isotropic} if $TY \subseteq \TYperp$ 
\item[] $TY$ is \emph{coisotropic} if $TY \supseteq \TYperp$ 
\item[] $TY$ is \emph{Lagrangian} if $TY = \TYperp$ 
\end{itemize}
\end{definition}
Hence, the dimension $k$ lies in the range $0 \leq k \leq n$ for isotropic subbundles, $n \leq k \leq 2n$ for coisotropic subbundles, and all Lagrangian subbundles are $n$ dimensional.  (Lagrangian subbundles exist only when the signature of $G$ is $(n,n)$.)  In the proof of proposition \ref{prop:stronger} it was shown that the submanifolds $Y$ corresponding to topological bibranes have the property that orthogonal vectors are also tangent to the submanifold; in our classification, they are coisotropic submanifolds.  Therefore, in the following we focus on coisotropic submanifolds (which include Lagrangian submanifolds as a special case).    

It is convenient to define a particular frame for $\TM$ adapted to the submanifold $Y$.  This is not quite as straightforward as in the Riemannian case, since a canonical splitting $\TM = TY \oplus \TYperp$ is no longer available.   However,  for coisotropic submanifolds there exists \cite{Bejancu:1995p230} a splitting of the tangent bundle of the form
\[
  \TM = \underbrace{\TYperp \oplus SY}_{TY} \oplus NY
\]
where $SY$ is a complementary \emph{screen} distribution to $\TYperp$ within $TY$ and $NY$ is a complementary \emph{transverse} distribution to $TY$ in $\TM$.  We have $\dim \TYperp = \dim NY = 2n-k$ and $\dim SY = 2k - 2n$.

Making use of this adapted frame, let us now write down an equivalent characterization of topological bibranes:
\begin{prop} \label{prop:coisotropic}
The pair $(Y,F)$ define a topological bibrane if and only if $Y \subseteq X \times \hX$ is a coisotropic submanifold such that $\ker F = \TYperp$ (i.e. the degenerate directions of $F$ and $G|_{TY}$ coincide) and, additionally,
\[
  (\tilde G^{-1} \tilde F)^2 = +1
\]
on $SY$, where $\tilde G \equiv G|_{SY}$ and  $\tilde F \equiv F|_{SY}$.    
\end{prop}
This is remarkably similar to the characterization of A-brane boundary conditions of the $(2,2)$ supersymmetric sigma model given in \cite{Kapustin:2001ij}, where the role of the antisymmetric symplectic form $\Omega$ on the target is exchanged for a neutral signature, symmetric metric $G$. Just as A-branes are required to be coisotropic with respect to $\Omega$, topological bibranes are required to be coisotropic with respect to $G$.  Moreover,  for A-branes the quotient bundle $TY / (TY)^\Omega$ is equipped with an endomorphism given by $\Omega^{-1} F$ and squaring to -1; similarly, for topological bibranes the quotient bundle $TY / \TYperp$ is equipped with an endomorphism $G^{-1} F$ squaring to +1. 

We consider two important special classes of topological bibranes.  When $F=0$, the condition that the degenerate directions of $F$ and $G|_{TY}$ agree implies that $TY = \TYperp$, i.e. $Y$ is a Lagrangian submanifold of $X \times \hX$.   This in turn implies that, locally, $Y$ is the graph of an isometry $f: X \to \hX$.   These \emph{graph-of-isometry} type bibranes are the  analogs of Lagrangian A-branes.   

The other special class of topological bibranes are the \emph{space-filling} bibranes with $Y = X \times \hX$.   In this case, $\TYperp$ is trivial and the screen distribution $SY$ is all of $TY$.  The condition in proposition \ref{prop:coisotropic}  then implies that $F$ is a symplectic form such that $(G^{-1} F)^2 = 1$.  Space-filling bibranes are the topological bibrane analog of space-filling A-branes.  

Indeed, if a manifold $M$ carries a symplectic form $F$, an almost product structure $R$, and a neutral signature metric $G$ with the compatibility requirement
\[
   F(u,v) = G(R u, v)
\]
for all vector fields $u,v$, and if in addition one demands that
\[
   dF = 0
\]
then $M$ is said to be an \emph{almost para-K\"ahler manifold} \cite{Alekseevsky:2009p273}.  Alternatively, the triplet of structures $(F,R,G)$ satisfying the above is known as an \emph{almost bi-Lagrangian structure} on the manifold \cite{Etayo:2006p347}, which terminology is inspired by the fact that the +1 and -1 eigenbundles of $R$ form two complementary Lagrangian subbundles with respect to the symplectic form $F$.  In case one of these subbundles is integrable, we call the manifold \emph{half para-K\"ahler} and in case both subbundles are integrable, we call the manifold  \emph{para-K\"ahler}.  Integrability of one of the eigenbundles need not imply integrability of the other eigenbundle.  Indeed, in section \ref{sec:susy01} we will require integrability of just the positive eigenbundle in order to ensure supersymmetry.   The study of para-K\"ahler manifolds is a rich and developing subbranch of indefinite signature geometry; they have appeared in an unrelated physical context in \cite{Cortes:2003zd}.

Another useful way to reformulate the topological bibrane condition is in terms of the language of \emph{generalized geometry} (see  \cite{Gualtieri:2007ng} for a review).  In the framework of generalized geometry one only works with the sum $TM \oplus T^* M$ as well as sections thereof (rather than $TM$ or $T^*M$ in isolation).  In particular, one speaks of \emph{generalized tangent vectors}  of $M$, defined as pairs $(u,\xi)$ with $u \in TM$, $\xi \in T^* M$.

The objects of ordinary of geometry have generalized counterparts; for instance, one defines the \emph{generalized tangent subbundle} of a submanifold $Y \subseteq M$ to be sum of its tangent bundle and conormal bundle (regarded as a subbundle of $T^*M$):
\[
  \tau Y = TY \oplus N^*Y  
\] 
Generalized geometry provides a prescription for incorporating a nonzero 2-form $F$:
\begin{definition}
The \emph{F-rotated generalized tangent subbundle} of a submanifold $Y \subseteq M$ is the set of generalized tangent vectors $(u,\xi)$ such that 
\[
  \tau^F Y = \{ (u,\xi): u \in TY, \xi \in T^*M \quad \text{such that $F u = e^*\xi$} \}
\]
Here $e:Y \to M$ is the embedding map and $F u$ represents the 1-form produced by contracting the 2-form $F$ on the vector $u$.
\end{definition}
What we have been calling allowed pairs of tangent vectors $(u,v)$ are nothing but sections of $\tau^F Y$ (after lowering the vector $v$ to a 1-form $\xi = G v$).

The generalized counterparts of complex structures are the \emph{generalized complex structures}: endomorphisms $\mathcal J : TM \oplus T^*M \to TM \oplus T^*M$ with $\mathcal J ^2 = -1$.  Symplectic manifolds with symplectic form $\Omega$ carry a generalized complex structure given by
\[
  \mathcal J_\Omega = \begin{pmatrix} & -\Omega^{-1} \\ \Omega & \end{pmatrix}
\]
(referring coordinates on $T^*M$ to a dual basis) and the A-branes are elegantly characterized as those $(Y,F)$ such that the F-rotated generalized tangent subbundle is stabilized by the $\mathcal J_\Omega$:
\begin{prop}(in reference \cite{Gualtieri:2007ng})
(Y,F) define an A-brane of a symplectic manifold if and only if  $\mathcal J_\Omega$ acting on a section of $\tau^F Y$  gives back another section of $\tau^F Y$, i.e.
\[
  \mathcal J_\Omega (\tau^F Y) = \tau^F Y
\]
\end{prop}
Motivated by our theme of comparing topological bibranes of the bosonic sigma model with A-branes of the $(2,2)$ sigma model by replacing $\Omega$ with the neutral signature $G$ and replacing almost complex structures with almost product structures, let us define a \emph{generalized (almost) product structure} to be an endomorphism $\mathcal R : TM \oplus T^*M \to TM \oplus T^*M$ with $\mathcal R ^2 = 1$.  Indeed, the generalized geometry object encoding the metric is a generalized almost product structure:
\begin{definition}
As defined in \cite{Gualtieri:2007ng}, the \emph{generalized metric}\footnote{Various authors disagree on what object is to be given the name ``generalized metric"; we follow the convention of \cite{Gualtieri:2007ng}.} is the particular generalized almost product structure given by
\[
  \mathcal{R}_G = \begin{pmatrix} & G^{-1} \\ G & \end{pmatrix}.
\]
\end{definition}  
Finally, we can compactly describe topological bibranes as follows.
\begin{prop}
\label{prop:generalizedmetric}
The pair $(Y,F)$ with $Y \subseteq M$ and $F$ a closed 2-form on $Y$, satisfy the conditions for a topological bibrane (as defined in the previous section) if and only if  $\mathcal R_G$ acting on a section of $\tau^F Y$  gives back another section of $\tau^F Y$, i.e.
\[
  \mathcal R_G (\tau^F Y) = \tau^F Y
\]
\end{prop}

\section{Topological defects of the (0,1) supersymmetric sigma model} \label{sec:susy01}

We now proceed to our main topic of interest: supersymmetric, topological defects of the (0,1) supersymmetric sigma model.   In addition to the bosonic fields $\phi^i$ and $\hphi^i$ we now have right-moving fermionic fields $\psip^i$ on $\Sigma$ and $\hpsip^i$ on $\hSigma$.

The bosonic gluing conditions described in section \ref{sec:bosonic} must be supplemented with fermionic gluing conditions on the fields $\psi^I = (\psi^i, \hpsi^i)$, setting half of the total fermionic degrees of freedom to zero the defect line $D$; we accomplish this by constraining $\psi^I$ to take values in an $n$-dimensional subbundle --  which we denote $\Rp$ -- of the $2n$-dimensional bundle $\TM$.    

Indeed, topological bibranes $(Y,F)$ are equipped with a natural middle dimensional subbundle, which we describe as follows
\begin{definition}
Let $(Y,F)$ be a topological bibrane, as characterized in propositions \ref{prop:topbibrane}, \ref{prop:coisotropic}, and \ref{prop:generalizedmetric}. Define $\Rp$ be the following bundle on $Y$:
\begin{equation}   \label{eq:rplus}
  \Rp  =  \{ u \in TY: G(u,w) = F(u,w) \quad \text{for all $w \in TY$} \}
\end{equation}
\end{definition}
Let us see what this subbundle corresponds to in the two special classes of topological bibranes we discussed previously.  In case $(Y,F)$ is space-filling, $\Rp$ is the +1 eigenbundle of the almost product structure $G^{-1} F$.  On the other hand, for graph-of-isometry type bibranes, this subbundle is simply the tangent bundle $TY$ itself.  
\begin{prop} \label{prop:rplus}
The subbundle $\Rp$, as defined above, is Lagrangian with respect to $\TM$.  In particular, it is $n$-dimensional.
\end{prop}
Given $(Y,F)$ satisfying the geometric topological bibrane conditions, we impose the bosonic gluing conditions written previously, and we augment these with condition that the fermions take values in $\Rp$; moreover, we require that $\Rp$ be an integrable distribution on $Y$.  For completeness, we record here the full set of gluing conditions:
\begin{align*}
\phi &\in Y \\
G(v,w) &= F(u,w) \\
G(s,w) &= F(s,w) 
\end{align*}
for all $w \in TY$, where $u^I \equiv \partial_0 \phi^I$, $v^I \equiv \partial_1 \phi^I$, and $s^I \equiv \psi^I$.

Let us now show that these gluing conditions define a topological, supersymmetry-preserving defect.   We write an explicit action for the $(0,1)$ supersymmetric sigma model with defect (see \cite{Melnikov:2003zv}):
\begin{align*}
  S &= \int_{\Sigma} d^2 \sigma \Big(-\frac{1}{2} g_{ij} \partial_\mu \phi^i \partial^\mu \phi^j 
  -\frac{1}{2} b_{ij} \epsilon^{\mu \nu} \partial_\mu \phi^i \partial_\nu \phi^j + \frac{i}{2} g_{ij} \psi_+^i D_- \psi_+^j \Big) \\
    &+ \int_{\hSigma} d^2 \sigma \Big(-\frac{1}{2} \hg_{ij} \partial_\mu \hphi^i \partial^\mu \hphi^j  -\frac{1}{2} \hb_{ij} \epsilon^{\mu \nu} \partial_\mu \hphi^i \partial_\nu \hphi^j + \frac{i}{2} \hg_{ij} \hpsi_+^i \hD_- \hpsi_+^j \Big) + \int_{D} d \sigma^0 \cA_I \partial_0 \phi^I
\end{align*}
where $\partial_\pm = \partial_0 \pm \partial_1$ and the covariant derivatives are given by
\begin{align*}
  D_{\pm} \psip^i &= \partial_{\pm} \psip^i + \partial_{\pm} \phi^j \Gamma^i_{jk} \psip^k \\
  \hD_{\pm} \hpsip^i &= \partial_{\pm} \hpsip^i + \partial_{\pm} \hphi^j \hat \Gamma^i_{jk} \hpsip^k
\end{align*}
Varying the action and picking out the term localized on $D$, one obtains
\[
  (\delta S)_D = \int_{D} d \sigma^0 \Big( -G_{IJ} \partial_1 \phi^I \delta \phi^J
        + F_{IJ} \partial_0 \phi^I \delta \phi^J  - \frac{i}{2} G_{IJ} \psip^I \delta^{\rm cov} \psip^J \Big) 
\]  
where $G$ and $F$ are as before and we have introduced the notation
\[ 
  \delta^{\rm cov} \psip^I = \delta \psip^I + \delta \phi^{J} \Gamma^I_{JK} \psip^K
\]
for the \emph{covariant variations}, which, unlike $\delta \psi^I$, transform covariantly under target space coordinate changes.   Here $\Gamma^I_{JK}$ are the coefficients of the Levi-Civita connection of the neutral signature metric $G$ on $X \times \hX$.

The first two terms in $(\delta S)_D$ vanish identically since the bosonic gluing conditions are identical to those in section \ref{sec:bosonic}.  Moreover, the third term in $(\delta S)_D$ can be shown to vanish by appealing to the following
\begin{prop} \label{prop:parallelness}
Let $(M,G)$ be a (pseudo-)Riemannian manifold with Levi-Civita connection $\nabla$.  Let $Y \subseteq M$ be a submanifold equipped with closed 2-form $F$ on its worldvolume, and let $\Rp$ be the subbundle defined in \eqref{eq:rplus}.  If $\Rp$ is integrable, then we have the following:
\[
  (e^*\nabla)_u s \in \Gamma(\Rp)
\]
for all $u \in \Gamma(TY)$ and $s \in \Gamma(\Rp)$.  Here $e^* \nabla$ is the pullback of $\nabla$ under the embedding map $e: Y \to M$.  In words, the subbundle $\Rp$ is $\nabla$-parallel with respect to $TY$.  
\end{prop}
\begin{prop} \label{prop:vanishing}
The term $G_{IJ} \psip^I \delta^{\rm cov} \psip^J$ vanishes for all allowed variations.
\end{prop}
Note that in case $F=0$, the subbundle $\Rp$  equals $TY$, and the above proposition reduces to the surprising result that the covariant derivative of sections of $TY$ are again in $TY$, i.e. Lagrangian submanifolds are automatically \emph{totally geodesic}, a result also obtained in \cite{Bejancu:1995p230}, \cite{Alekseevsky:2009p273}.  

We now check that the above gluing conditions define a topological defect and that they preserve (0,1) supersymmetry of the theory.  The bosonic portions of the stress-energy tensors automatically glue smoothly by the analysis of section \ref{sec:bosonic}.   It remains to check that the contributions involving fermions glue smoothly; i.e.
\[
  (T^{\, \mu}_{\;\;\; \nu})^f - ({\widehat T}^{\,\mu}_{\;\;\; \nu})^f = 0
\]  
for all $\mu, \nu$.  This is equivalent to the condition
\[
  G(\nabla_u s, s) = 0
\]
where the vector fields $u,s$ are as we have defined previously.  This equation is satisfied, once again by appeal to our parallelness proposition.

Finally, the condition for preservation of supersymmetry is that the 1-component of the supercurrents glue smoothly.  This is equivalent to
\[
  G(s, t) = 0
\]
where $s \equiv \psip$ and $t \equiv \partial_+ \phi$.   This is satisfied by virtue of the fact that both  $\partial_+ \phi$ and $\psip$ are constrained to take values in $\Rp$ by the gluing conditions, and the fact that $\Rp$ is a Lagrangian subbundle.  

\section{Topological defects of the $(0,2)$ supersymmetric sigma model} \label{sec:susy02}

Consider now the case when the target spaces $X$ and $\hX$ are K\"ahler manifolds with K\"ahler forms $\omega$ and $\homega$, complex structures $j$ and $\hj$.  In this case the supersymmetric sigma models on either side of the defects possess $(0,2)$ supersymmetry; let us geometrically classify those topological bibranes $(Y,F)$ that preserve both supercharges.

We impose the gluing conditions written in the previous section and assume that the conditions for $(0,1)$ supersymmetry are met.  Then $(0,2)$ supersymmetry will follow if
\[
  G'^1 - \widehat G'^1 = 0
\]
along $D$, where $G'$ and $\widehat G'$ are the second supercurrents.  This is the condition
\begin{equation} \label{eq:omega}
   \Omega(s,t) = 0
\end{equation}
for all $s,t \in \Rp$, where
\[
  \Omega = G J
\]
is the K\"ahler form on $X \times \hX$ with respect to complex structure $J = j \oplus \hj$.   We can say the following about when \eqref{eq:omega} is satisfied:
\begin{prop}  
The following conditions on topological bibrane $(Y,F)$ are equivalent:
\begin{enumerate}
\item $\Rp$ is an $\Omega$-Lagrangian subbundle.
\item $J(\Rp) = \Rp,$ i.e. the leaves of the foliation by subbundle $\Rp$ are complex submanifolds with respect to complex structure $J = j \oplus \hj$.
\item $Y$ is a complex submanifold of $X \times \hX$ and $F$ has type $(1,1)$ with respect to the induced complex structure on $Y$.
\item  In generalized geometry terms, $\mathcal J_J$ acting on a section of $\tau^F Y$  gives back another section of $\tau^F Y$, where 
\[
  \mathcal J_J = \begin{pmatrix} J &  \\ & -J^\T \end{pmatrix}
\]
is the generalized complex structure associated with complex structure $J$.  
\item $(Y,F)$ is a B-brane with respect to complex structure $J$.
\end{enumerate}
\end{prop}
The generalized metric $\mathcal{R}_G$ and the generalized complex structures $\mathcal J_\Omega$ and $\mathcal J_J$ are related by
\[
  \mathcal{R}_G = -\mathcal J_\Omega \mathcal J_J
\]
This shows that if the bundle $\tau^F Y$ is stabilized by any two of these, then it is automatically stabilized by the third.  In particular, we have the following
\begin{prop} 
\label{prop:02defects}
The bibrane $(Y,F)$ is a topological, supersymmetry preserving defect of the $(0,2)$ sigma model when it is both an A-brane with respect to symplectic form $\Omega = \omega \oplus -\homega$ and a B-brane with respect to complex structure $J = j \oplus \hj$ on $X \times \hX$ (with suitable integrability conditions).
\end{prop}

\section{Topological defects and T-duality}  \label{sec:tduality}

We briefly comment on the role our topological defects play in implementing dualities of the bosonic sigma models with targets $X$ and $\hX$.   For concreteness let $X$ and $\hX$ be flat tori and let $Y = X \times \hX$ be a space-filling defect equipped with a line bundle with curvature $\cF = d\cA$; we define the 2-form $F = -\cF - B$, as in section 2.

Topological defects separating theories $X$ and $\hX$ can be fused with those separating $\hX$ and $X'$ to yield topological defects separating $X$ and $X'$.  The invisible defect separating $X$ and $X$ (corresponding to the diagonal $Y \subset X \times X$ and $F=0$) is a unital element with respect to this operation of fusion and an \emph{invertible defect} is one that can be fused with another defect to yield the invisible defect.  On general grounds, an invertible topological defect implements a duality\footnote{See \cite{Davydov:2010rm} for a recent discussion of this point in the context of rational CFTs.}.

Hence, we expect that if the 2-form $\cF$ is chosen such that $(Y,F)$ is an invertible topological defect, then the sigma model with target $X$ will be related to the sigma model with target $\hX$ by a duality transformation.  For tori, the group of duality transformations is $O(n,n;\ZZ)$ (the group of automorphisms of a lattice of signature $(n,n)$) \cite{Giveon:1994fu}.

Let us fix the background data $(g,b)$ and take
\[
  \cF = \begin{pmatrix} f & h \\ -h^\T & \hf  \end{pmatrix}.
\]
where $f, \hf$ are antisymmetric $n \times n$ matrices.  Since $\cF$ is the curvature of a line bundle on the torus $X \times \hX$, its entries are constrained to be integers.  Generically, the determinant of the off-diagonal matrix $h$ can take any integer value; however, let us now restrict the form of $\cF$ by assuming that $\det h = \pm 1$, or equivalently $h \in GL(n,\ZZ)$.  

The topological defect condition for space-filling defects states that $(G^{-1} F)^2 = 1$; writing this out in terms of the blocks, this is equivalent to the following relationship between the background data $(g,b)$ for $X$ and $(\hg, \hb)$ for $\hX$:
\begin{align*}
\hg &= h^\T \big( g - (b+f) g^{-1} (b+f) \big)^{-1} h \\
0 &= (b+f) g^{-1} h + h \hg^{-1}(\hb - \hf) 
\end{align*}
We recognize this as the duality transformation consisting of a shift of the B-field by the matrix $f$, followed by a T-duality (inverting the data $g+b$), followed by a basis change of the lattice generating the torus (parametrized by unimodular matrix $h$), followed by another B-field shift (this time by $\hf$) (see \cite{Giveon:1994fu}, eqs.~(2.4.25),~(2.4.26),~(2.4.39)).  The corresponding duality element is
\[
  \begin{pmatrix} 
    \hf h^{-1} & h^\T + \hf h^{-1} f \\
    h^{-1} & h^{-1} f
  \end{pmatrix}
  \in O(n,n;\ZZ)
\] 
(A similar calculation with $f = \hf = 0$ and $h=1$ is performed in \cite{Sarkissian:2008dq}.)  It is evident from this expression that quantization of the entries is implied by quantization of the entries of $\cF$.

\section{Conclusions} \label{sec:conclusions}

We have provided several equivalent characterizations of topological defects of the bosonic sigma model (``topological bibranes") in propositions \ref{prop:topbibrane},  \ref{prop:coisotropic}, and \ref{prop:generalizedmetric}: on the one hand, topological bibranes are closely related to the A-branes of symplectic geometry (exchanging the symplectic form for a neutral signature metric); on the other hand, one can provide an elegant description of topological bibranes in terms of the generalized metric on the product of the targets.  Dimension $n$ topological bibranes correspond to graphs of isometries $X \to \hX$, and dimension $2n$ topological bibranes are equipped with a (half) para-K\"ahler structure, i.e. a neutral signature metric $G$ together with an endomorphism $R=G^{-1} F$ squaring to the identity, such that $R^\T G R = -G$.

Fermionic fields are incorporated by requiring the fermions to take values in a certain middle dimensional subbundle $\Rp$ with which topological bibranes are equipped.  We have shown the key result that this subbundle is automatically parallel with respect to all of $TY$ (assuming integrability of $\Rp$ and $dF = 0$). 

We remark that the class of topological defects studied in this paper is not broad enough to close under the operation of fusion, due to the fact that we have not allowed our defects to include extra degrees of freedom living on the interface.  For this reason, we leave a detailed study of fusion to future work.  

Another question is to what extent the preservation of the various symmetries (which we have analyzed here on the classical level) carry over to the quantum level.  Presumably, one can argue that the invertible defects are exactly supersymmetric and topological, but it would be interesting to find out which other of our defects have this property.  

The present work was motivated in part by a desire to understand whether something analogous to D-branes can be defined for the heterotic string.  Since the worldsheet of the heterotic string possesses only right-moving fermions, a worldsheet-with-boundary is ruled out, but perhaps not a worldsheet-with-defect.  It would be quite interesting to study applications of the present analysis of defects of the bosonic, (0,1), and (0,2) sigma models to heterotic string theory.  

\appendix

\section{Proofs of selected propositions} \label{app:proofs}

In this appendix we offer proofs of selected propositions discussed in the text.
\begin{prop*}[\ref{prop:stronger}]
If $(Y,F)$ are such that $G(u,u) = -G(v,v)$ for all allowed pairs $(u,v)$, then $v \in TY$ and $(v,u)$ is an allowed pair as well.
\end{prop*}
\begin{proof}
First, we show that $Y$ is necessarily a coisotropic submanifold (in the terminology of section \ref{sec:geometry}).  This is because $v_0 \in \TYperp$ implies that $(0,v_0)$ is an allowed pair.  Then
\[
  G(v_0, v_0) = -G(0,0) = 0
\]
This, in turn, implies that $\TYperp$ is an isotropic subspace, since
\[
  G(v_0, v_0 ') = \frac{1}{2} G(v_0 + v_0', v_0 + v_0') = 0
\]
for all $v_0, v_0' \in \TYperp$.  That is to say, $\TYperp \subseteq (TY^\bot)^\bot = TY$, which is what it means for $Y$ to be coisotropic.  

Next, we show that if $(u,v)$ is an allowed pair, then $v$ is in the subspace $TY$.  (The definition of allowed pair requires $u \in TY$ but not, a priori, $v \in TY$.)   Let $v_0 \in \TYperp$.  For any value of a real parameter $t$, it is easy to see that  $(u, v+ t v_0)$ is also an allowed pair.  Hence,
\[
  -G(u,u) = G(v+ t v_0, v + t v_0)
\]
for all $t$.  Differentiating both sides of this equation with respect to $t$, we have
\[
  0 = \frac{d}{dt} G(v + t v_0, v + t v_0) = 2 G(v, v_0)
\]
(since $G(v,v)$ is $t$ independent and $G(v_0, v_0) = 0$).  The choice of $v_0 \in \TYperp$ was arbitrary, so
\[
  v \in (\TYperp)^\bot = TY
\]
Finally, we show that $(u,v)$ an allowed pair implies $(v,u)$ is an allowed pair as well.  By linearity, if $(u,v)$ and $(u', v')$ are allowed pairs, then $(u+u', v+v')$ is also an allowed pair.  By assumption,
\[
  G(u+u', u+u') = -G(v+v',v+v') 
\]
which, using $G(u,u) = -G(u',u')$ and $G(u', u') = -G(v', v')$, implies that
\begin{equation} \label{eq:Guv}
  G(u, u') = -G(v, v') .
\end{equation} 
Let $(u,v)$ be an allowed pair and $w \in TY$ be arbitrary.  It is not hard to see that there exists a $v' \in TY$ such that $(w,v')$ is an allowed pair.  Setting $u' = w$ in \eqref{eq:Guv} above, one finds that
\[
G(u,w) = -G(v,v') = -G(v', v) = -F(w,v) = F(v,w)
\] 
We conclude that $(v,u)$ is also an allowed pair.
\end{proof}

\begin{prop*}[\ref{prop:coisotropic}] 
The pair $(Y,F)$ define a topological bibrane if and only if $Y \subseteq X \times \hX$ is a coisotropic submanifold such that $\ker F = \TYperp$ (i.e. the degenerate directions of $F$ and $G|_{TY}$ coincide) and, additionally,
\[
  (\tilde G^{-1} \tilde F)^2 = +1
\]
on $SY$, where $\tilde G \equiv G|_{SY}$ and  $\tilde F \equiv F|_{SY}$.    
\end{prop*}
\begin{proof}
To prove ``only if,'' suppose that $(Y,F)$ is a topological bibrane.  Theorem \ref{prop:stronger} tells us that if $(u,0)$ is an allowed pair then $(0,u)$ is an allowed pair as well, and vice versa;  this implies that $\ker F = \TYperp$.   Hence, the degenerate directions of $F$ and $G|_{TY}$ coincide and $\tilde F$ and $\tilde G$ (the restrictions of $G$ and $F$ to the screen distribution $SY$) are both nondegenerate.  

Moreover, if $\tilde u, \tilde v \in SY$, then by definition, $(\tilde u, \tilde v)$ is an allowed pair if and only if $\tilde v = (\tilde G^{-1} \tilde F) \tilde u$.  Since $(u, (\tilde G^{-1} \tilde F) \tilde u)$ is an allowed pair,  $((\tilde G^{-1} \tilde F) \tilde u,u )$ is an allowed pair as well (again, by theorem \ref{prop:stronger}).  This means that $\tilde u = (\tilde G^{-1} \tilde F)^2 \tilde u$ and since $\tilde u \in SY$ was arbitrary, we have
\[
   (\tilde G^{-1} \tilde F)^2 = 1 
\]
To prove ``if," let $(u,v)$ be an allowed pair; we wish to show that $v \in TY$ and $(v,u)$ is an allowed pair as well.

Let $v_0 \in \TYperp$ be arbitrary.  We then have
\[
  G(v, v_0) = F(v, v_0) = 0
\]
where the first equality is by the definition of allowed pair and the second equality is by the fact that $\ker F = \TYperp$.  Hence $v \in (\TYperp)^\bot = TY$. 

Now, we write $u = u_0 + \tilde u$ and $v = v_0 + \tilde v$ for the decomposition of $TY$ vectors with respect to the splitting $TY = \TYperp \oplus SY$.   Since $F(u,w) = G(v,w)$ for all $w \in TY$, we have $F(\tilde u, \tilde w) = G(\tilde v,\tilde w)$ for all $\tilde w \in SY$, or $\tilde v =  (\tilde G^{-1} \tilde F) \tilde u$.  But $ (\tilde G^{-1} \tilde F)^2 = 1$, so $\tilde u =  (\tilde G^{-1} \tilde F) \tilde v$.  Hence, $F(\tilde v, \tilde w) = G(\tilde u,\tilde w)$ for all $\tilde w \in SY$, and finally $F(v,w) = G(u,w)$ for all $w \in TY$.
\end{proof}
\begin{prop*}[\ref{prop:rplus}] 
The subbundle $\Rp$, as defined above, is Lagrangian with respect to $\TM$.  In particular, it is $n$-dimensional.
\end{prop*}
\begin{proof}
It is convenient to first refine the description of the tangent bundles of topological bibranes given in proposition \ref{prop:coisotropic}.  Let $\tilde R$ be the endomorphism $\tilde G^{-1} \tilde F: SY \to SY$, which by proposition \ref{prop:coisotropic}, squares to the identity.  Hence the screen distribution admits a splitting $SY = \tilde {\Rp}  \oplus \tilde{\mathcal R_-}$ into the +1 and -1 eigenbundles of $\tilde R$.  Indeed, since $\tilde G^{-1} \tilde F$ is the product of a symmetric and an antisymmetric matrix, it is traceless, and in particular, exactly half of its eigenvalues are +1 and half are -1.  Hence,
\[
  \dim \tilde{\Rp} = \dim \tilde{\mathcal R_-} = \frac{1}{2}\dim SY = \frac{1}{2}(2k - 2n) = k-n
\]
where $\dim M = 2n$ and $\dim Y = k$ (recall that since $Y$ is coisotropic, $k \geq n$).  From the definition of the bundle $\Rp$, it follows straightforwardly that
\[
  \Rp = \TYperp \oplus \tilde{\Rp}
\] 
Hence, its dimension is
\[
  \dim \Rp =  \dim \TYperp + \dim \tilde{\Rp} =  (2n-k) + (k-n) = n
\]
Moreover, if $s,t \in \Rp$, then
\[
  G(s,t) = F(s,t) = -F(t,s) = -G(t,s) = -G(s,t) = 0
\]
so $\Rp$ is an isotropic subbundle of $\TM$.  (In the first and third equalities above, we have used the definition of $\Rp$  and in the second and fourth equalities we have used the symmetries of $F$ and $G$.)  Since it is middle dimensional and isotropic, it follows that $\Rp$ is Lagrangian with respect to $\TM$.   
\end{proof}

\begin{prop*} [\ref{prop:parallelness}]
Let $(M,G)$ be a (pseudo-)Riemannian manifold with Levi-Civita connection $\nabla$.  Let $Y \subseteq M$ be a submanifold equipped with closed 2-form $F$ on its worldvolume, and let $\Rp$ be the subbundle defined in \eqref{eq:rplus}.  If $\Rp$ is integrable, then we have the following:
\[
  (e^*\nabla)_u s \in \Gamma(\Rp)
\]
for all $u \in \Gamma(TY)$ and $s \in \Gamma(\Rp)$.  Here $e^* \nabla$ is the pullback of $\nabla$ under the embedding map $e: Y \to M$.  In words, the subbundle $\Rp$ is $\nabla$-parallel with respect to $TY$.  
\end{prop*}
\begin{proof}
We will use the involutivity of $\Rp$ (which follows from integrability) to establish the formula
\begin{equation} \label{eq:dFust}
2 G(  (e^*\nabla)_u s, t) = dF(u,s,t)  
\end{equation}
for all $u \in \Gamma(TY)$ and $s,t \in \Gamma(\Rp)$, keeping in mind that we regard sections of $TY$ and $\Rp$ as sections of the pullback bundle $\TM$, equipped with pullback connection $e^*\nabla$.  Since $dF=0$ by assumption, this formula will imply that 
\[
   (e^*\nabla)_u s \in \Gamma(\Rp ^\bot) =  \Gamma(\Rp)
\]
In the space-filling case, we have noted that our manifold is (half) para-K\"ahler, and this formula reduces to  an equation appearing in Theorem 1 in \cite{Cortes:2003zd}; the proof will follow in a similar vein.

To show \eqref{eq:dFust}, we first write down the Kosczul formula
\[
  2 G(  (e^*\nabla)_u s, t) = u(G(s,t)) + s(G(u,t)) - t(G(u,s)) 
  + G([u,s],t) - G([u,t],s) - G([s,t],u)
\]
Since $s,t \in \Gamma(\Rp)$. we can replace $G$ by $-F$ in all of the terms on the right hand side above except for the last one.  However, as $\Rp$ is involutive, $[s,t] \in \Rp$, so that we can replace $G$ by $+F$ in the last term:
\[
  2 G(  (e^*\nabla)_u s, t) = -u(F(s,t)) - s(F(u,t)) + t(F(u,s)) 
  - F([u,s],t) + F([u,t],s) - F([s,t],u)
\]
On the right hand side, we have the right form for the exterior derivative of $F$, except for the sign of the first term $-u(F(s,t))$.  However, note that $\Rp$ is an isotropic subbundle with respect to both $G$ and $F$, so this term is zero and we can flip its sign.  The equation \eqref{eq:dFust} is proved.  
\end{proof}

\begin{prop*}[\ref{prop:vanishing}]
The term $G_{IJ} \psip^I \delta^{\rm cov} \psip^J$ vanishes for all allowed variations.
\end{prop*}
\begin{proof}
We write the fermionic boundary conditions as 
\[
  {R(\phi)^I}_J \psip^J = \psip^I
\]
where $R:\TM \to \TM$ is a $\phi$-dependent matrix squaring to the identity.   Varying the above and writing in terms of the covariant variations, one has
\[
  (\nabla_K {R^I}_J) \delta \phi^K \psip^J + {R^I}_J \delta^{\rm cov} \psip^J = \delta^{\rm cov} \psip^I
\]
or, since $\delta \phi$ is tangent for allowed variations and $\psip$ takes values in $\Rp$,
\[
   (\nabla_K {R^I}_J) {\pi^K}_L {P^J}_M \delta \phi^L \psip^M +  {R^I}_J \delta^{\rm cov} \psip^J = \delta^{\rm cov} \psip^I
\]
where $\pi$ is a projector onto $TY$ and $P = \frac{1}{2}(1+R)$ is the projector onto $\Rp$.   However, by proposition \ref{prop:parallelness}, the first term vanishes.  We then have
\[
  G_{IJ} \psip^I \delta^{\rm cov} \psip^J = (G_{IJ} {P^I}_K {P^J}_L) \psip^K \delta^{\rm cov} \psip^L = 0
\]
since $\Rp$ is an isotropic subbundle.
\end{proof}


\end{document}